%% file: GOSPAL.tex
\documentclass[sigconf]{acmart}
\usepackage{amsmath,amssymb}
\usepackage{bbm,bm}
\usepackage{float}
\usepackage{epsfig}
\usepackage{graphics}
\usepackage{graphicx}
\usepackage{epstopdf}
\usepackage{subcaption}
\usepackage{textcomp}
\usepackage{booktabs} 
\usepackage{algorithm}
\usepackage{calrsfs}
\usepackage{times} 
\usepackage{tikz}
\usepackage{algpseudocode}
\usepackage{float}
\newtheorem{definition}{Definition}
\newtheorem{theorem}{Theorem}[]

\newtheorem{lemma}[]{Lemma}

\def\cS{\mathcal{S}}
\def\cN{\mathcal{N}}
\def\cX{\mathcal{X}}
\def\bmr{\bm{r}}
\def\bmx{\bm{x}}
\def\bmq{\bm{q}}
\def\bmp{\bm{p}}
\def\tU{\tilde{U}}
\def\*{\star}

\setcopyright{rightsretained}

\acmDOI{10.475/123_4}

\acmISBN{123-4567-24-567/08/06}

\acmConference[VALUETOOLS'19]{ACM VALUETOOLS Conference}{March 2019}{Spain}
\acmYear{2019}

\editor{Jennifer B. Sartor}
\editor{Theo D'Hondt}
\editor{Wolfgang De Meuter}

\begin{document}

\title{GOSPAL: An Efficient Strategy-Proof Mechanism for Constrained Resource Allocation}

\author{Indu Yadav}
\affiliation{%
  \institution{Dept. of Electrical Engineering \\ IIT Bombay, India}
}
\email{indu@ee.iitb.ac.in}

 \author{Prasanna Chaporkar}
\affiliation{%
   \institution{Dept. of Electrical Engineering\\ 
   IIT Bombay, India}
}
 \email{chaporkar@ee.iitb.ac.in}

 \author{Abhay Karandikar}
\authornote{A.Karandikar is currently the director of IIT Kanpur (on leave from IIT Bombay) email: karandi@iitk.ac.in.}
 \affiliation{ %
 \institution{Dept. of Electrical Engineering,\\ IIT Bombay, India\\ 
     Director and Professor, IIT Kanpur}
     }
     \email{karandi@ee.iitb.ac.in}

\begin{abstract}
 We consider allocation of a resource to multiple interested users with a constraint that if the resource
is allocated to user $i$ then it can not be allocated simultaneously to a predefined set of users $\cS_i$. This 
scenario arises in many practical systems that include wireless networks and constrained queuing systems. It is known that
socially optimal strategy-proof mechanism is not only NP-hard, but it is also hard to approximate. This
renders optimal policy hard to use in practice. Here, we propose a computationally efficient mechanism
and prove it to be strategy proof. Using Monte Carlo simulations, we show that the social utility 
of the proposed scheme is close to that of the optimal. Further, we demonstrate how the proposed 
mechanism can be used for fair and efficient short-term spectrum allocation in resource constrained large wireless networks.
\end{abstract}
\keywords{ Strategy-Proof, Resource Allocation, Mechanism Design, Wireless networks}
\maketitle

\input{introduction-2}

\input{system_model}
\input{gospal_algo}
\input{example}

\input{simulation}


%

\section{Conclusion}
\label{sec:conclusion}
In this paper, we consider resource allocation problem among multiple users with constrained set. We propose Grouping based Optimal Strategy-Proof Allocation (GOSPAL) Algorithm based on the grouping of non-conflicting users. We prove that GOSPAL is strategy-proof. Since the proposed algorithm is computationally efficient,  it is feasible to implement even in large number of users. Using simulation results, we observe that GOSPAL achieves social utility and resource utilization close to the optimal. The mechanism also achieves better fairness index for resource allocation among the users in comparison to the other existing scheme.

\bibliographystyle{ACM-Reference-Format}
\bibliography{refer}

\end{document}

%% file: introduction-2.tex
\section{Introduction}
Optimal use of resources has always been a concern. In today's
world with increasing demand, efficient use of limited and scarce
resource has become a challenge. A resource allocation mechanism
must be designed to meet the strategic goals of the system. The
desired strategic goals may include maximization of social welfare,
efficient and fair utilization of limited resource and maximization of
revenue. Auction based mechanism is a popular way of distributing
the available resource among users \cite{krishna2009auction}.
Here, it is assumed that each
user has some quantitative valuation of the usefulness of the resource
for them. Based on the valuation, users bid for the resource and the
auctioneer or centralized controller arbitrates resource distribution based
on the received bid values. Each user is assumed to be greedy
and rational. Hence, a key challenge in using auctions for resource
distribution is that the users may not bid their true valuation if that
benefits them. Hence, key challenge in this approach is to device a
strategy-proof/truthful auction in which users do not have any incen-
tive in not bidding their true valuation \cite{osborne1994course}.
Hence, being rational, their bids would equal their respective true valuations. A well known
Vickrey Clarke Grooves (VCG) auction provides a framework for
designing strategy proof mechanism that maximizes social utility
 \cite{vickrey1961counterspeculation,clarke1971multipart,groves1973incentives}. 
However, depending on scenario under consideration
VCG may not be computationally feasible. Thus, alternative ways
are required.

In this paper, we consider a constrained resource allocation prob-
lem in which a single resource can be allocated to multiple interested
users. Users are assumed to be selfish and rational. The allocation is
constrained in a sense that corresponding to each user $i$, there exists
a pre-specified set of users $\cS_i$ such that if the resource is given to
user $i$, then it can not be given to any user in $\cS_i$. This constrained
resource allocation is relevant in many applications including sched-
uling in ad-hoc networks, overlay device-to-device communication
and spectrum allocation in cellular networks. Moreover, in many of
these scenarios the resource in allocated for some time after which
the reallocation is done, e.g. a channel is allocated to non-interfering
links in wireless network for a predefined time slot and the allocation
may vary from slot to slot based on many factors including backlog,
delay, data priority and channel state. Thus, we assume that the
resource is auctioned periodically. It is known that VCG auction is
NP-hard for constrained resource allocation \cite{garey2002computers}. Thus, it can not be
used it this case as resource allocation needs to be computationally
fast for repeated execution.

In many applications for which the aforementioned constrained
auction mechanism is relevant, there are desirable features other than
maximization of social welfare for an allocation mechanism. For
example, if we consider each user $i$ to be a Base Station (BS) in a
cellular system and  $\cS_i$  to be the set BSs that interfere with $i$, then
for efficient utilization of the scarce and expensive spectrum must
be ensured. Thus, the number of BSs that get the spectrum may also
be important. Also, in the case where BSs may belong to different
operators, it is possible that \textquotedblleft small\textquotedblright  operators may not be able to
bid comparatively with \textquotedblleft big\textquotedblright  operators and hence may never get the
spectrum. It may be desirable to make the auction based allocation
fair in a sense that fraction of time slots in which the BSs get the
spectrum must be balanced across the BSs. {\it Our aim in this paper is
to propose a strategy-proof and computationally feasible auction for
resource allocation for constrained resource allocation problem. We
show that the proposed mechanism not only achieve near-optimal social utility
but also perform well in terms of resource utilization
and fairness parameters.}

Now, we state some most relevant previous work for the problem under consideration. 
As stated before VCG auctions are applicable
for the problem under consideration, but it is computationally infeasible.
To address this, computationally efficient strategy-proof 
alternatives that include greedy \cite{zhou2008ebay} and SMALL \cite{wu2011small} are proposed. 
While SMALL has poor resource utilization, greedy suffers from
unfair allocation when considered across time. 
In \cite{trust} authors propose a double auction based spectrum allocation  mechanism which holds strategy-proof property.
Authors in \cite{gopinathan} propose spectrum allocation mechanism for repeated auctions for maximizing the social welfare. 
They also ensure fairness in spectrum allocation among the users.
Some other relevant work includes the algorithms for real time dynamic spectrum allocation among the base station \cite{gandhi2008towards,subramanian2008near}. These algorithms are not strategy-proof.
Auction based general framework for spectrum allocation in cognitive network has been proposed by the authors in \cite{kasbekar}.
Authors in \cite{ji2007cognitive} propose auction based approach is efficient for spectrum allocation among the base stations in cognitive networks. Authors in \cite{matinmikko2014spectrum} present Licensed Shared Access framework for spectrum sharing between incumbent and secondary users. This framework requires efficient mechanism to address dynamic load variations in the  network.
A comprehensive overview of the game
theoretic approaches used for spectrum sharing in cognitive radio
networks is given in \cite{ji2007cognitive}.


{Our main contributions are as follows:}
\begin{itemize}
	\item We study the problem of resource allocation among multiple users with limited constrained set, which is NP-Hard and becomes intractable for large set of users.
	\item We propose truthful mechanism Group Optimal Strategy-Proof Allocation (GOSPAL) for resource allocation in generic framework and also present its application in resource allocation among the base stations in wireless networks.
	\item With simulations we demonstrate that the proposed mechanism performs well in terms of achieving the social efficiency of resource allocation close to the optimal solution. The proposed mechanism is also computationally efficient, which makes it practical feasible for implementations in the large set of users.
\end{itemize}

The rest of the paper is organized as follows. Section \ref{sec:sys_model} describes the system model and problem formulation. In Section \ref{sec:algo}, proposed algorithm GOSPAL is presented in detail. We demonstrate the applicability of the proposed algorithm in spectrum allocation in wireless networks in Section \ref{sec:application}. Simulations results are presented in Section \ref{sec:Simulation}.In Section \ref{sec:conclusion} we conclude our work.

%% file: system_model.tex
\section{System Model and  Preliminaries}
\label{sec:sys_model}
 We consider the resource allocation problem where a resource can be allocated to the multiple users. The time is slotted, and each slot
 is called allocation frame. The framework for resource allocation comprises of a centralized controller, database and a set of users. The database contains the information about the resource available for allocation. The centralized controller is the key entity which performs resource allocation among the users
 in each allocation frame. Each user $i$ has a constraint set $\mathcal{S}_i$ 
 such that if resource is allocated to user $i$ then it can not be allocated to the users contained in the $\cS_i$ simultaneously. We assume that the constraints are symmetric, i.e. if $j \in \cS_i$ then $i \in \cS_j$. Let $\mathcal{N} = \{1,2,\ldots,n\}$ denote the set of users. We assume all the users, though rational, act selfishly and are non co-operative. 
 The following steps are repeated in every frame $t$:
 
%
%
%
Each user in the system has a hidden state associated with it. Depending on its state, users make their demand valuation. Let $r_i(t)$ denote the actual demand valuation of the user $i$ in frame $t$. The actual demand valuation is a user's private information. Each user $i$ communicates {\it bid} value $q_i(t)$ to the centralized controller at the beginning of frame $t$, and $q_i(t)$ need not be equal to $r_i(t)$. Based on the bids received, the controller allocates resources while respecting the constraints. The allocation remains for the frame duration. In the following, we discuss resource allocation problem that needs to
be solved in every allocation frame. For brevity, we omit $t$ from notation.
   
\subsection{Action Space and Action Profile}
 Action profile of a user is the set of all possible bids that can be chosen by the user. Let $q_i$ denote the bid of user $i$. Here, a user is free to declare any positive $q_i$. Therefore, the action profile of user $i$ is $q_i \in \mathcal{R}_{+}$. In case any user has bid 0, then it is {\it not} considered for resource allocation. Hence, without
 loss of generality, we assume that all users $q_i > 0$.
 Action space of the system with $n$ users is denoted as $n$-tuple, $\bm{q} = (q_1,\ldots,q_n)$ comprising of the bids corresponding to each user. Since action profile of each user is positive real value, the action space of $n$ users correspond to $\bmq \in \mathcal{R}_{+}^{n}$. Also, define the action profile $\bmq_{-i}\in \mathcal{R}_{+}^{n-1}$ comprising of bids from all the users as in $\bmq$ except that of user $i$. Next, we formulate the problem that needs to be
 solved in each allocation frame.

\subsection{Problem Formulation}
\label{sec:prob_form}
Let $\bmx = (x_1,\ldots,x_n)$ denote a indicator binary vector indicating resource allocation such that
$x_i = 1$ only when the resource is allocated to user $i$.  
\begin{definition} \label{defi:feasible_rsa}
	A binary vector $\bmx$ is a feasible resource allocation if $x_i + \sum_{j \in \cS_i}x_j \leq 1$
	for every $i\in\cN$. Moreover, if  $x_i + \sum_{j \in \cS_i}x_j = 1$
	for every $i\in\cN$, then $\bmx$ is called maximal resource allocation.
\end{definition} 
Let $\cX$ denote the set of all feasible resource allocations.
Next we define auction based resource allocation mechanism.
\begin{definition}\label{defi:rss_allocation}
	An auction based resource allocation policy $\pi$ is a map from $\mathcal{R}^n_+$ to $\cX \times [0,\infty)^n$, i.e. for given bids $\bmq$, $\pi$ outputs feasible allocation 
	$\bmx^\pi(\bmq)$ and a price vector 
	$\bmp^\pi(\bmq) = (p_1^\pi(\bmq),\ldots,p_n^\pi(\bmq))$.
\end{definition}
Thus, a resource allocation policy $\pi$ outputs a feasible resource allocation for any given bid vector $\bmq$,
and also the price that each user needs to pay for the allocated resource.
Let $\Pi$ denote the set of all auction based allocation policies.
\begin{definition} \label{defi:social_utility}
	Social utility under $\pi$ for bid values $\bmq$ is defined as 
	$U_s^\pi(\bmq) = \sum_{i=1}^{n} r_i x_i^\pi(\bmq)$. Moreover, utility for user $i$ for bids $\bmq$ 
	under $\pi$ is given 
	as $U_i^\pi(\bmq) = (r_i - p_i^\pi(\bmq)) x_i^\pi(\bmq)$. 
\end{definition}

Note that the value of social utility is the sum of true evaluations, not the bid values, of the users
to which $\pi$ allocates resources. Moreover, utility for a user is the difference between its true
evaluation $r_i$ and price $p_i$ charged under policy $\pi$. Aim of the centralized controller is
to design $\pi$ that maximizes social utility, i.e. it wants 
\begin{align}\label{eqn:pi*}
\pi^\* \in \arg\max_{\pi \in \Pi} U_s^\pi(\bmq),
\end{align}  
while each user wants to bid so as to maximize its own utility. Note that $\bmr$ is the private information with the users, and the centralized controller may not know it. Thus, we need to design mechanism in which rational users have no incentive to submit bid other than its true evaluation.
\begin{definition}\label{def: strategy-proofness}
	A mechanism $\pi$ is truthful (strategy-proof) if 
	\[U^\pi_i(r_i,\bmq_{-i}) \geq U^\pi_i(\bmq), \mbox{\ for all $\bmq\in \mathcal{R}_{+}^{n}$}. \]
\end{definition}
Thus, for strategy-proof mechanism $\pi$, users have no incentive to bid anything other than its true 
evaluation. Next, for completeness sake, we describe VCG mechanism that obtains strategy proof optimal
mechanism $\pi^\*$ given in (\ref{eqn:pi*}).

\subsection{Vickrey-Clarke-Grooves (VCG) Mechanism}
\label{sec:VCG}
VCG mechanism for the specified problem is implemented in two steps as described next.
In the first step, given bid vector $\bmq$,
find $U^\*(\bmq) = \max_{\bmx \in \cX} \sum_{i=1}^{n} q_i x_i$. Thus, $\bmx^\*(\bmq)$ is a resource allocation that maximizes the total bid value aggregated over all users to which resource is allocated.
In the second step, price for each user is computed as described next. Let $\cX_{-i}$ denote the set of all
possible resource allocations over set of users $\cN \setminus \{i\}$ with no change in the constraint sets
for the considered users. Let $U^\*_{-i}(\bmq) = \max_{\bmx \in \cX_{-i}} \sum_{j\not=i} q_j x_j$.
Now, price for user $i$ is given by \[p_i = U^\*_{-i}(\bmq) - (U^\*(\bmq)-q_i).\]

The VCG mechanism is known to be strategy proof. Hence, users do not have any incentive to bid anything other than $r_i$. Thus, 
VCG allocation solves the desired optimization problem (\ref{eqn:pi*}). In spite of optimality, VCG mechanism is not practical as it is NP-hard to
compute resource allocation that achieves $U^\*(\bmq)$ and prices $\bmp$. 
Thus, finding polynomial optimal resource allocation is not feasible unless P=NP, and hence
computationally feasible sub-optimal resource allocation strategy with desired characteristic are required.
Next, we propose one such algorithm.

%% file: gospal_algo.tex
\section{GOSPAL Algorithm}
\label{sec:algo}
 In this section, we describe an efficient strategy-proof mechanism for resource allocation. The mechanism is implemented in two phases: (1)~Resource Allocation phase and (2)~Pricing phase. The resource allocation phase arbitrates users that should get the resource. If a user is allocated resource, 
 how much it must to pay is decided in a pricing phase. The design of a pricing scheme is essential to enforce strategy-proofness (truthfulness). First, we describe the resource allocation strategy and the pricing scheme is presented subsequently.
 
 \subsection{Resource Allocation Phase}
 As discussed above, the resource allocation phase determines which users should be allocated the resource, i.e. we select $\bmx \in \cX$. Towards this end, our first step is to randomly partition the set of all
 users $\cN$ into at most $\eta$ non-conflicting groups denoted as $\{G_1,\ldots,G_\eta \}$, 
 where $\eta = \max_{i\in \cN} |\cS_i|+1$. 
 Here, $|A|$ denotes the cardinality of set $A$. The partitioning is achieved using an iterative greedy algorithm. In first iteration a user is selected at random, and it is put in group $G_1$. In a typical 
 iteration, a user $i$ is picked at random from $\cN \setminus \cup_{k=1}^\eta G_k$ and it is put in the
group $G_{u_{\min}}$ such that $u_{\min} = \min\{u : u\in\{1,\ldots,\eta\} \mbox{  and } G_u \cap \cS_i  = \phi\}$.
We continue this process until $\cup_{k=1}^\eta G_k = \cN$. Pseudo code for the randomized 
conflict-free grouping is provided in Algorithm~\ref{alg:rcg}.
\begin{algorithm}[t]
	\caption{Pseudo code for randomized  conflict-free grouping
		\label{alg:rcg}}
	\begin{algorithmic}
		\State \textit {\bf Input:} $\cN$, $\cS_i$ for every $i \in \cN$
	\end{algorithmic}
	\begin{algorithmic}
		\State \textit{\bf Output:} A conflict-free partition $\{G_1,\ldots,G_\eta\}$
	\end{algorithmic}
	\begin{algorithmic}[1]
		\State Initialize $\cN_{\rm temp} = \cN$ and $G_u = \phi$ for every $u=1,\ldots,\eta$ \label{step:initialize}
		\While{$\cN_{\rm temp} \not= \phi$}
		\State Choose a user, say $i$, from $\cN_{\rm temp}$ uniformly at random \label{step:11}
		\State Find $u_{\min} = \min\{u: u\in\{1,\ldots,\eta\} \mbox{  and } G_u \cap \cS_i = \phi  \}$ \label{step:12}
		\State $G_{u_{\min}} \leftarrow G_{u_{\min}} \cup \{i\}$
		\EndWhile
	\end{algorithmic}
\end{algorithm}
Following lemma summarizes key properties of the partitioning step.
\begin{lemma}\label{lemma:rcg}
	A conflict-free grouping algorithm given in Algorithm~\ref{alg:rcg} outputs a partition $\{G_1,\ldots,G_\eta\}$
	of $\cN$ such that if $i,j \in G_u$, then $j \not\in \cS_i$.  
\end{lemma}
\begin{proof}
	We need to show that the RHS in Step~\ref{step:12} of the algorithm is a non-empty set in every iteration. Rest follows immediately from how this set is constructed. Required follows from the fact maximum
	cardinality of any $\cS_i$ is $\eta - 1$. Thus, there exist at least one $u$ such that $G_u \cap \cS_i = \phi$.
\end{proof}

The Lemma~\ref{lemma:rcg} states that the resource can be allocated to all the members of any group
$G_u$ without violating the allocation constraint. Moreover, it is important to note that the grouping
does not depend on the bid values $\bmq$.

Now,  let $\Omega_g$ denote the set of all possible orderings of the sets $\{G_1,\ldots,G_\eta\}$ obtained using  conflict-free grouping algorithm. Thus, $|\Omega_g| = \eta!$. Furthermore,
let $\omega_j \in \Omega_g$ denote the $j^`{th}$ ordering of the groups in the set $\Omega_g$.
We denote $\omega_j$ by a tuple $(G_{j1},\ldots,G_{j\eta})$.
For example if $\eta=3$, then there are $|\Omega_g| = 3! = 6$ different orderings. 
One of the possible $6$ group ordering or tuple is $ \omega_j = (G_2, ~G_1, ~G_3)$. 
Thus, $G_{j1} = G_2$, $G_{j2} = G_1$ and $G_{j3} = G_3$.
A resource allocation given group ordering $\omega_j$ is done as follows.
We first assign the resource to each user in $G_{j1}$, then to all the users in 
$G_{j2} \setminus (\cup_{i \in G_{j1}} \cS_i)$, and so on. Pseudo-code to obtain
resource allocation corresponding to group ordering $\omega_j$ is given in
Algorithm~\ref{alg:ralo}.
\begin{algorithm}[t]
	\caption{Pseudo code for resource allocation for given group ordering $\omega_j$
		\label{alg:ralo}}
	\begin{algorithmic}
		\State \textit {\bf Input:} $G_{ju}$ for every $1\leq u \leq \eta$, $\cS_i$ for every $i \in \cN$
	\end{algorithmic}
	\begin{algorithmic}
		\State \textit{\bf Output:} A resource allocation $\bmx(j)$
	\end{algorithmic}
	\begin{algorithmic}[1]
		\State Initialize $G_{\rm temp} = \phi$, $\ell = 1$ and $x_i(j) =0$ for all $i \in \cN$  \label{step:initialize2}
		\While{$\ell \leq \eta$}
		\State $G_a \leftarrow G_{j\ell} \setminus (\cup_{i \in G_{\rm temp}} \cS_i)$
		\State $x_i(j) \leftarrow 1$ for every $i \in G_a$ \label{step:21}
		\State $G_{\rm temp} \leftarrow G_{\rm temp} \cup G_a$ \label{step:22}
		\State $\ell \leftarrow \ell + 1$
		\EndWhile
	\end{algorithmic}
\end{algorithm}
Following guarantee can be given about output of the algorithm.
\begin{lemma}
	\label{lemma:res_aloc}
	The resource allocation vector $\bmx(j)$ given by Algorithm~\ref{alg:ralo} corresponding to any
	group tuple $\omega_j$ is feasible, i.e. $\bmx(j) \in \cX$.
	Moreover, $\bmx(j)$ is maximal allocation vector for every $j$.
\end{lemma}
\begin{proof}
	Let $\bmx_{\ell}(j)$ denote the allocation after $\ell$ iterations of the algorithm. We first show that
	$\bmx_{\ell}(j) \in \cX$ for every $1\leq \ell \leq \eta$. Note that for $\ell =1$,
	$x_{\ell i } = 1$ only for $i\in G_{j1}$. From Lemma~\ref{lemma:rcg}, $\bmx_1(j) \in \cX$ follows.
	Suppose $\bmx_{\ell}(j) \in \cX$ holds for every $1 \leq \ell \leq \ell'$. Consider $(\ell' +1)^{\rm th}$
	iteration of the algorithm. Note that the $G_{\rm temp}$ in every iteration contains users to which
	the resource is allocated until that iteration. Note that in Step~\ref{step:22} of the algorithm the
	resource is allocated only to users in $G_{j(\ell'+1)}$ that do not conflict with the users in $G_{\rm temp}$. This proves that $\bmx_{\ell'+1}(j) \in \cX$ and the required follows using induction. Now, we prove that
	the resource allocation is maximal. Suppose not, then there exist a user $u$ such that $x_{\ell u}(j) = 0$ in the output of the algorithm, but $\bmx'$ such that $x'_i = x_{\ell i}(j)$ for every $i \not= u$ and $x_u' = 1$ is in $\cX$.
	Since, $(G_{j1},\ldots,G_{j\eta})$ is a partition of $\cN$, $u$ must belong to some $G_{j\ell}$. Also,
	$u$ must not belong to $\cS_i$ for any $i$ which is allocated the resource in first $\ell - 1$ iterations of the algorithm. But, then the algorithm will allocate resource to user $u$ in $\ell^{\rm th}$ iteration. Hence, no such user exists. This proves the required.
\end{proof}

Now define, with little abuse of notation, the perceived social utility under allocation $\bmx(j)$ as
\[ \tU_j(\bmq)   = \sum_{i=1}^{n} q_i x_i(j) . \]
Moreover, define $j^\*_{\bmq} = \arg\max_{\{j: \omega_j \in \Omega_g  \}} \tU_j(\bmq) $. Thus, $\omega_{j^\*}$ is the group permutation for which perceived utility is maximized among all possible group permutations. We propose to choose resource allocation
$\bmx(j^\*_{\bmq})$. Note that even though the grouping does not depend on the bids $\bmq$, the chosen resource allocation does. 
Let $\tU^\*(\bmq)$ denote the maximum value of the perceived social utility  for the bids $\bmq$.
Next, we describe our proposed pricing scheme.

\subsection{Pricing Scheme}
 \par After the resource allocation, we propose the appropriate pricing scheme which ensures the strategy-proofness of the proposed algorithm. That is, if any user tries to deviate from its actual demand, it is penalized. 
 Let $\bmq_{-i}(\epsilon)$  denote the bid vector in which the bids of all the users except $i$ are same as that in $\bmq$, but the bid of
 user $i$ is $\epsilon > 0$ in $\bmq_{-i}(\epsilon)$.
 Now, the price charged from the user $i$  is given as:
 \begin{equation}
  \label{eqn:price_gospal}
  p_i(\bmq) =  \left[\lim_{\epsilon \downarrow 0} \tU^\*(\bmq_{-i}(\epsilon)) - (\tU^\*(\bmq) - q_i)\right] \times x_i(j^\*_{\bmq} ).
 \end{equation}
 
 We state the following straightforward result.
 
 \begin{lemma}
 	\label{lemma:p_prop}
 	Under any bid values $\bmq > 0$, $0 \leq p_i \leq q_i$ for every $i \in \cN$.
 \end{lemma}
 \begin{proof}
	Note that for every $\epsilon > 0$, \[\tU^\*(\bmq_{-i}(\epsilon)) \geq \tU^\*(\bmq) - q_i + \epsilon. \] 
	Thus the required follows by taking limit $\epsilon \downarrow 0$ on both sides of the above inequality.
\end{proof}

This lemma clearly shows that for any truthful user $i$, utility obtained is non-negative irrespective of the bids of other users.
 
 Note that the optimal group permutation under bid vectors $\bmq$ and $\bmq_{-i}(\epsilon)$ can be different. Also, note that
 unlike VCG, in our pricing scheme we do not completely remove user $i$, rather user $i$ is always present. Only the bid value
 of user $i$ goes to zero. This distinction is important as removing a user changes resource allocation conflicts. As illustration consider 
 a system with five users with constraint sets given by $\cS_1 = \{3,4,5  \}$, $\cS_2 = \phi$, $\cS_i = \{1\}$ for $i=3,4,5$.
 Suppose grouping given by Algorithm~\ref{alg:rcg} is $G_1 = \{1,2\}$ and $G_2 = \{3,4,5\}$. Consider permutation
 $\omega_1 = (G_1,G_2)$. Thus, as per Algorithm~\ref{alg:ralo}, the resources will be first allocated to all the users in $G_1$ and then to the
 users in $G_2$ that do not have conflict with the users in $G_1$. Note that when user 1 is present in the system, resource can not be 
 allocated to any user in $G_2$. But, if we remove user 1 completely, then all the user in $G_2$ can get the resource. Thus, there is a clear
 difference under our scheme with regard to when user is present and when it is not. It is important to note that the price in our scheme is
 calculated while retaining user in the system unlike VCG. Pseudo code for the proposed algorithm
 is given in Algorithm~\ref{alg:gospal}. Next we prove the key properties of our proposed algorithm.

 \begin{algorithm}[t]
 	\caption{Pseudo code for GOSPAL mechanism}
 		\label{alg:gospal}
 	\begin{algorithmic}
 		\State \textit{\bf Input:} bid vector $\bmq$, $\cS_i$ for every $i \in \cN$
 	\end{algorithmic}
 	\begin{algorithmic}
 		\State \textit{\bf Output:} Resource allocation $\bmx(\bmq)$ and price vector $\bm{p}(\bmq)$
 	\end{algorithmic}
 	\begin{algorithmic}[1]
 		\State Use Algorithm~\ref{alg:rcg} to obtain conflict free grouping $(G_1,\ldots,G_\eta)$
 		\For{$\omega_j\in \Omega_g$}
 		\State Find allocation $\bmx(j)$ using Algorithm~\ref{alg:ralo}
 		\State Compute $\tU_j(\bmq)   = \sum_{i=1}^{n} q_i x_i(j) $
 		\EndFor
 		\State Find $j^\*_{\bmq} = \arg\max_{\{j: \omega_j \in \Omega_g  \}} \tU_j(\bmq) $
 		\State Choose $\bmx(\bmq) = \bmx(j^\*_{\bmq} )$
 		\State Compute prices using (\ref{eqn:price_gospal})
 	\end{algorithmic}
 \end{algorithm}

\begin{lemma}(\textbf{Monotonicity})
\label{lemma:1}
 If a user $i$ is allocated resources for bids $\bmq$, then it will also be allocated resources for bids $\bmq_{-i}(\epsilon)$ for every $\epsilon > q_i$.
 Moreover, optimal group permutation under $\bmq$ and $\bmq_{-i}(\epsilon)$ are the same, i.e. $j^\*_{\bmq} = j^\*_{\bmq_{-i}(\epsilon)}$.
\end{lemma}
\begin{proof}
	Without loss of generality, let $\epsilon = q_i + \Delta$ for some $\Delta > 0$. Note that since the bid value of only user $i$ has changed,
	we can conclude that 
	\begin{align}
	\label{eq:l11}
	\tU_j(\bmq_{-i}(\epsilon)) - \tU_j(\bmq) \le \Delta,
	\end{align}
	 for every group permutation $\omega_j$. Moreover,
	 \begin{align}
	 \label{eq:l12}
	 \tU_{j^\*_{\bmq}}(\bmq) + \Delta = \tU_{j^\*_{\bmq}}(\bmq_{-i}(\epsilon)),
	 \end{align}
	 i.e. the perceived social utilities under group permutation 
	$j^\*_{\bmq}$ for bid vectors $\bmq$ and $\bmq_{-i}(\epsilon)$ differ by amount $\Delta$ with latter having the larger value. Thus,
	we can conclude from (\ref{eq:l11}) and (\ref{eq:l12}) that $j^\*_{\bmq}$ is optimal group permutation for $\bmq_{-i}(\epsilon)$  as well.
	Now, the required follows from Algorithm~\ref{alg:ralo}.
\end{proof}

Lemma~\ref{lemma:1} implies that if a user unilaterally increases its bid, then it is more likely to get the resources. 
Next, we prove that our proposed algorithm is strategy proof. 

\begin{theorem}Algorithm \ref{alg:gospal} is strategy-proof.
\end{theorem}

\begin{proof}
	We proof the required by considering two scenarios.
	\textit{\textbf{Scenario 1}} : User $i$ bids more than its true valuation, i.e. $q_i > r_i$. Without loss of generality, $q_i = r_i + \Delta$
	for some $\Delta > 0$. Bids of the other users can be arbitrary. Thus, we compare two bid vectors, viz. $\bmq$ and $\bmq_{-i}(r_i)$,
	where latter corresponds to user $i$ bidding truthfully.
	This scenario is further bifurcated into three cases.\\

	\noindent \textit{Case (i)}: User $i$ gets resource under both bid vectors $\bmq$ and $\bmq_{-i}(r_i)$. By Lemma~\ref{lemma:1},
	it follows that the optimal group permutation remains same for both the bid vectors.  It follows that the optimal perceived utility values satisfy
	$\tU^\*(\bmq) = \tU^\*(\bmq_{-i}(r_i)) + \Delta$. Now,  from (\ref{eqn:price_gospal}), it follows that $p_i(\bmq) = p_i(\bmq_{-i}(r_i))$.
	Thus the required holds.
	
	\noindent \textit{Case (ii)}: User $i$ does not get the resource under $\bmq_{-i}(r_i)$, but gets it under $\bmq$. Note that utility for user 
	$i$ under $\bmq_{-i}(r_i)$ is zero as it does not get the resource. Now, we bound user $i$ utility under $\bmq$. Since 
	the bid for only user $i$ is different under two bid vectors, we can conclude that
	\begin{align} \label{eqn:l13}
	\tU^\*(\bmq)  - \tU^\*(\bmq_{-i}(r_i))  \leq \Delta. 
	\end{align}
	Now, from (\ref{eqn:price_gospal}), it follows that
	\begin{align}
	p_i(\bmq) & = \lim_{\epsilon \downarrow 0} \tU^\*(\bmq)  - (\tU^\*(\bmq) - q_i) \nonumber \\
	& = \lim_{\epsilon \downarrow 0} \tU^\*(\bmq) - (\tU^\*(\bmq) - r_i) + \Delta  \label{eqn:l14} \\
	& =  (\tU^\*(\bmq_{-i}(r_i)) - \tU^\*(\bmq)) + r_i + \Delta \label{eqn:l15} \\
	& \geq r_i. \label{eqn:l16} 
	\end{align}
	
	Equality (\ref{eqn:l14}) follows as $q_i = r_i + \Delta$. Equality (\ref{eqn:l15}) follows by Lemma~\ref{lemma:1}. Note that for every $\epsilon$
	smaller than $r_i$ user $i$ can not get resource as it can not get it  at bid value $r_i$. Moreover, since only bid for user $i$ is changing, the
	optimal perceived social utility remains unchanged. Hence, the limiting value equals maximum perceived social utility for bid $\bmq_{-i}(r_i)$.
	Finally, (\ref{eqn:l16}) follows from (\ref{eqn:l13}). Now, (\ref{eqn:l16}) implies that the utility for user $i$ under $\bmq$ can at most be 0, which is same when it bids true valuation $r_i$. This proves the required.
	
	\noindent\textit{Case (iii)}: The user $i$ neither gets resource at actual demand $q_i$, nor at $q_i + \Delta$. Here, utility for 
	user $i$ will remain zero. \\
	
	\noindent \textit{\textbf{Scenario 2}} : User $i$ bids less than its true valuation, i.e. $q_i < r_i$. Without loss of generality, $r_i = q_i + \Delta$
	for some $\Delta > 0$. Bids of the other users can be arbitrary. Thus, we compare two bid vectors, viz. $\bmq$ and $\bmq_{-i}(r_i)$,
	where latter corresponds to user $i$ bidding truthfully.
	This scenario is further bifurcated into three cases.\\
	
	\textit{Case (i)}: The user $i$ is allocated resource under $\bmq_{-i}(r_i)$ and also under $\bmq$. Analysis of this case is similar to that in Case~(i) of
	Scenario~1. Again here, it can be shown that the utility for user remains unchanged, and hence there is no benefit for deviating from true evaluation.
	
	\textit{Case (ii)}: The user $i$ is allocated resource under $\bmq_{-i}(r_i)$, but it does not get it under $\bmq$. This implies that the user  $i$ 
	has utility  $r_i - p_i(\bmq_{-i}(r_i))$ for bid vector $\bmq_{-i}(r_i)$, but on deviation its utility becomes zero. Now, the required follows from Lemma~\ref{lemma:p_prop}. 
	
\textit{Case (iii)}: The user $i$ neither gets a resource at $\bmq_{-i}(r_i)$ nor at $\bmq$. Here, the utility for the user remains Thus, no incentive on deviation from actual demand. This completes the proof.
\end{proof}

Computational complexity of the proposed algorithm is $O((\eta!)^2 n^2)$.
Recall that $\eta$ is the maximum cardinality of the constraint sets. In many applications, constraint set cardinality does not increase with the
number of users, e.g. cellular systems. In such systems, the algorithm provides computationally efficient strategy-proof mechanism for resource allocation.
In the next section, we describe the functioning of the proposed algorithm using example.

%% file: example.tex
\section{Illustrative Example}
\label{sec:application}
In this section, we describe a practical application in which the proposed mechanism can be applied. 
We consider the problem of channel (spectrum) allocation in wireless networks. Due to limited availability of spectrum, 
efficient and fair allocation of spectrum becomes essential. Spectrum allocation can be performed using auctions. We describe application of the GOSPAL algorithm for the channel allocation among multiple base stations using sealed bid auctions. 
We consider a geographical region with $n$ base stations deployed to provide services to the subscribers. 
We assume that the spectrum database contains information about the channel available for allocation.
Let us assume that each base station provides coverage to the disc of radius $R$ units centered at the base station location. 
This means that, if any two base stations that are placed distance less than $R$ units can potentially interfere with each other 
whenever they are allocated the same channel. In other words, the base stations $B_1$ and $B_2$ are said to be interfering pairs if the
distance between them is less than $R$ units. 
Thus, in our model, base stations are the users. 
For a base station $i$, $\cS_i$ is the set of all base stations that are less than $r$ units away from $i$. 
Note that to avoid interference, a channel can be given to at most one user in $\{i\} \cup \cS_i$.
Typically in a wireless network, base stations are organized to cover certain geographical area and hence $|\cS_i|$ is
expected to be small for each $i$. 

We assume that the base stations compete for the channel. The need for resource depends on current state at the base station. For example,
if the base station has a large backlog, or if it had to serve time critical application, then its need for the resource
 is more in order to achieve the desired quality of service. Let $r_i$ denote the true evaluation at the base station $i$.
 To obtain the resource, base station $i$ submits bid $q_i$. Depending on the bids received, the centralized controller decides on which
 base stations get the channel. Channel allocation needs to be interference free.
Next, we demonstrate spectrum allocation in a wireless network using GOSPAL mechanism.

\par \textbf{Illustrative Example:} Consider a network consisting of $6$ base stations distributed across the region as illustrated in Figure \ref{network}. We assume each base station submits a non-zero bid to the auctioneer. The base stations in the wireless network are denoted as node and the interfering pair of base stations have an edge joining them.  Let the bid vector $\bmq = [5 ~7 ~8 ~9 ~6 ~9]$. 
We describe resource allocation under three strategy-proof resource auction mechanisms, one proposed here and two from the literature.
Specifically, we consider SMALL and greedy auction mechanisms proposed in \cite{wu2011small} and \cite{zhou2008ebay}, respectively. 
\begin{figure}[h]
\centering
\includegraphics[width = 0.25\textwidth]{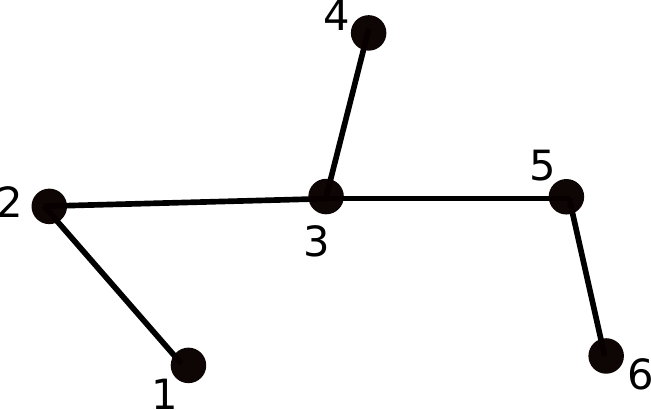}
\caption{Wireless Network}
\label{network}
\end{figure}

\begin{table}[h]
	\caption{Base Station Grouping }\label{tab:grp_list}
	\centering
	\begin{tabular}{| l || l| } 
		\hline
		
		\multicolumn{2}{|c|}{Group List}\\
		\hline
		$G_1$ & \{1, 5\} \\ 
		\hline
		$G_2$ & \{2, 4\}  \\ 
		\hline
		$G_3$ & \{3, 6\} \\ 
		\hline
	\end{tabular}
\end{table}
\vspace{-0.3cm}
\begin{table}[h]
	\caption{Group arrangements \& social welfare}\label{tab:permutation}
	\centering
	\begin{tabular}{|p{1.5cm}|| p{1.9cm} || p {1.5cm}|} 
		\hline 
		Sequence No. &Group arrangements & Social Welfare\\
		\hline
		\hline 
		1 & [$G_1$,  $G_2$,  $G_3$] & 20\\ 
		\hline
		2 & [$G_1$,  $G_3$,  $G_2$] & 20\\ 
		\hline
		3 & [$G_2$,  $G_1$,  $G_3$] & 22 \\ 
		\hline
		4 & [$G_2$,  $G_3$,  $G_1$] & 25 \\ 
		\hline
		5 & [$G_3$,  $G_1$,  $G_2$] & 22 \\ 
		\hline
		6 & [$G_3$,  $G_2$,  $G_1$] & 22 \\ 
		\hline
		
	\end{tabular}
\end{table}

\subsection{Allocation in GOSPAL}
In the first step, GOSPAL partitions the set of all users into non-conflicting groups using Algorithm~\ref{alg:rcg}.
Let the output of Algorithm~\ref{alg:rcg} be as shown in Table~\ref{tab:grp_list}.
With $3$ groups, we get $|\Omega_g| = 6$ arrangements. 
For each arrangement, we determine the social welfare $\tU(\bmq)$ based on the bids. 
Table~\ref{tab:permutation} provides the social utility for each permutation (see Algorithm~\ref{alg:ralo}).
Note that the permutation $[G_2,G_3,G_1]$ achieves the highest social utility 25. As per
Algorithm~\ref{alg:ralo}, the resource allocation chosen by GOSPAL is $\bmx^\*(\bmq) = (0,1,0,1,0,1)$ (see Algorithm~\ref{alg:gospal}).

\subsection{Allocation in SMALL}
Auction mechanism SMALL also partitions the set of all users randomly into non-conflicting groups.
Without loss of generality, let the partition be same as that shown in Table~\ref{tab:grp_list}.
SMALL determines the group valuation $\sigma(G_j)$ for each group $j$, where
$\sigma(G_j) = (|G_j| -1)\times \min\{q_j | j \in G_j \}$. Each user, except the one with the least bid, in the group with maximum evaluation 
gets the channel. The users that get the channel pay price equal to the least bid value in the group.
In the example considered here, $G_3$ has the highest group evaluation of 8. Hence, only user 6 gets the channel and it
pays price equal to 8.

\subsection{Allocation in Greedy}
The greedy scheme starts by picking a user with the highest bid value. This user gets the channel. In the iterative step, 
the highest bidding user in the set of users which are neither already selected nor are in the constraint set of the selected users is selected.
The process continues until no user can be selected. The price for a selected user $i$ is calculated as follows:
We remove the user $i$ from the system. Compute the greedy allocation for the remaining $n-1$ users. Let $c$ be the 
highest bid value of of a selected user in $\cS_i$ in the new allocation. Then, the price for $i$ is $c$. In the example considered here,
users 6, 4 and 2 get the channel and the prices for these are 6, 8 and 5, respectively.  

In the constructed example, GOSPAL and greedy have the same social utility, while SMALL achieves much lesser value.
Moreover, GOSPAL and greedy provide maximal allocation, while SMALL has lesser resource utilization. Though the proposed algorithm
performs on par with existing schemes in the constructed example, to understand the performance
comparison of these schemes we perform Monte Carlo simulation as described in the following section.

%% file: simulation.tex
\section{Simulation Results}
\label{sec:Simulation}
In this section, we compare the performances of various strategy-proof mechanisms for resource allocation.
To model the resource allocation constraints, we generate a random undirected graph $G=(V,E)$ whose nodes represent users and 
the constraint set for user $i$ equals the set of neighbors in $G$. The random graphs are generated with the 
desired degree distribution using configuration model \cite{chung2002connected}. In all our simulations, the maximum degree is restricted to 4.
Performance of the proposed scheme is compared against VCG, SMALL and greedy schemes.
Simulations are performed in MATLAB \cite{moler1982matlab}. 
We compare the performances based on three parameters:

\noindent
$\bullet$ Social Welfare: It is defined as the sum of the valuations of the base stations which are assigned channels. \\
\noindent
 $\bullet$ Spectrum Utilization: It is defined as the total number of base stations which are assigned channels in allocation phase. \\
 \noindent
$\bullet$ Fairness across time: It quantifies disparity between the average number of times the channel is allocated to various base stations.

First, we compare the performance of the proposed mechanism, SMALL and greedy with VCG for small graph sizes (up to 21 nodes). At each node, bids are generated
at random in the interval [5,15]. Figure \ref{compare_1}, shows the social welfare and and the spectrum utilization obtained under the four schemes.
Note that VCG based allocation provides better (optimal) in both respect to the ones in GOSPAL, SMALL and greedy. However, computation of the VCG allocation is computationally challenging task even for sparse modest size network. On the other hand, GOSPAL, SMALL and greedy can be used to provide resource allocation for large networks. Note that our proposed algorithm outperforms SMALL in all the cases significantly. More importantly, resource utilization under our scheme is much better than both SMALL and greedy, which is close to that in VCG. 

\begin{figure*}[t!]
    \centering
    \begin{subfigure}[t]{0.28\textwidth}
        \centering
       \includegraphics[width=\textwidth]{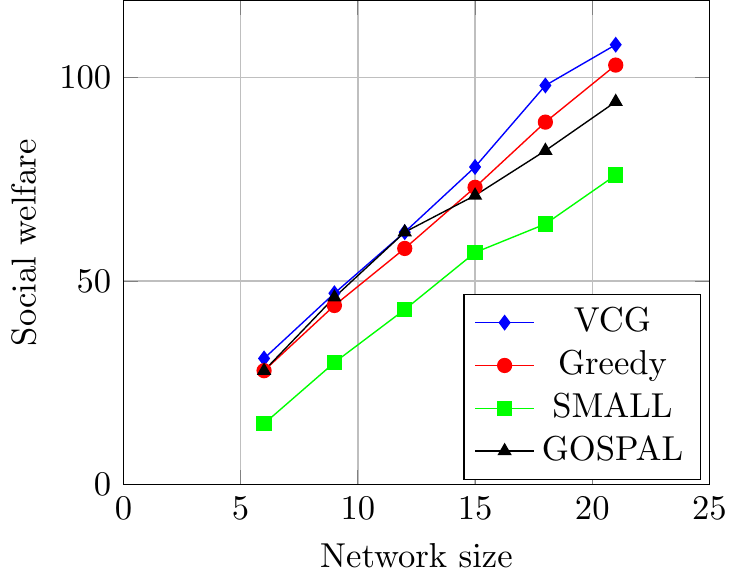}
        \caption{{Social Welfare}}
        \label{fig:sw1}
        \end{subfigure}%
    ~
    \begin{subfigure}[t]{0.28\textwidth}
        \centering
       \includegraphics[width=\textwidth]{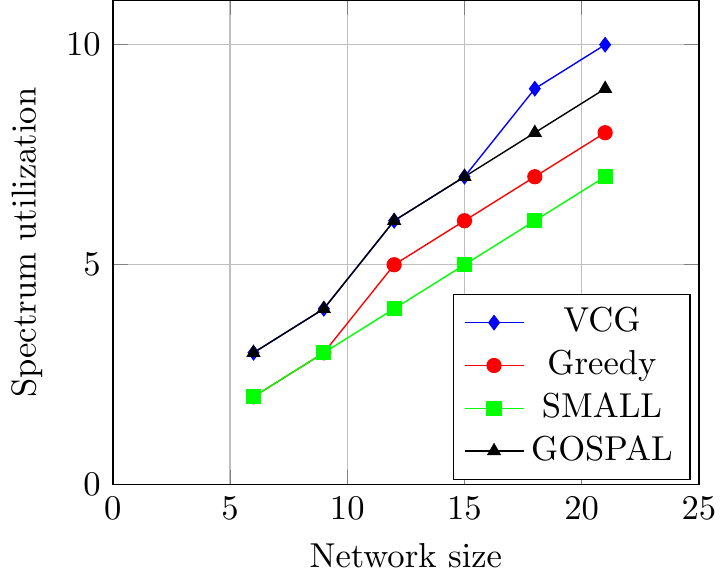}
        \caption{{Spectrum utilization}}
        \label{fig:su1}
        \end{subfigure}
     ~
         \begin{subfigure}[t]{0.28\textwidth}
        \centering
       \includegraphics[width=\textwidth]{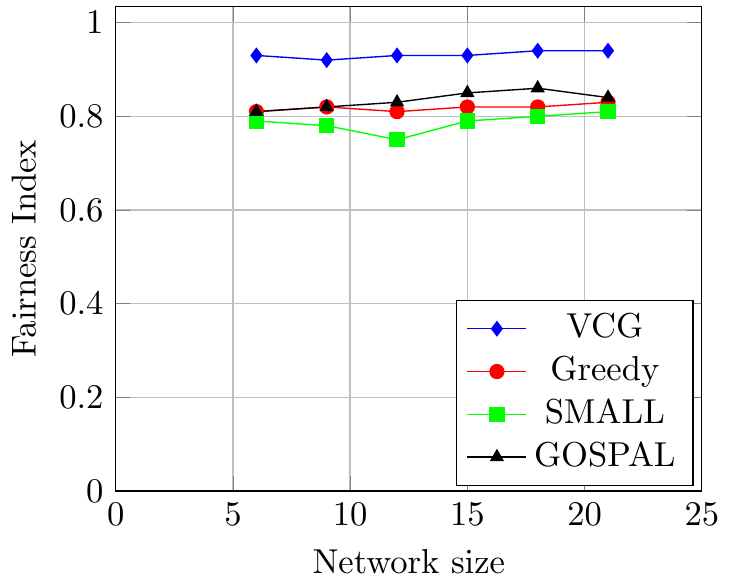}
        \caption{{Fairness Index}}
        \label{fig:fair_idx1}
        \end{subfigure}
\caption{{Performance comparison for different algorithms in small network.}}
\label{compare_1}
\end{figure*}

To further understand the performance of the proposed algorithm, we perform simulations on large networks. In this case, because of computational 
intractability of VCG allocation, we compare the results of our scheme with SMALL and greedy resource allocation mechanisms.
In Figure~\ref{compare_2}, we consider networks in which degree distribution is uniform over $\{1,2,3,4\}$. The base station bids are uniformly distributed
in the interval $[8,30]$.
The results shown are averaged over 100 different topologies with bids chosen independently for each base station. It can be observed that greedy mechanism provides the highest value of social welfare among all the schemes. However, GOSPAL provides marginally better spectrum utilization. Both these scheme
significantly outperform SMALL.

\begin{figure*}[t!]
    \centering
    \begin{subfigure}[t]{0.28\textwidth}
        \centering
       \includegraphics[width=\textwidth]{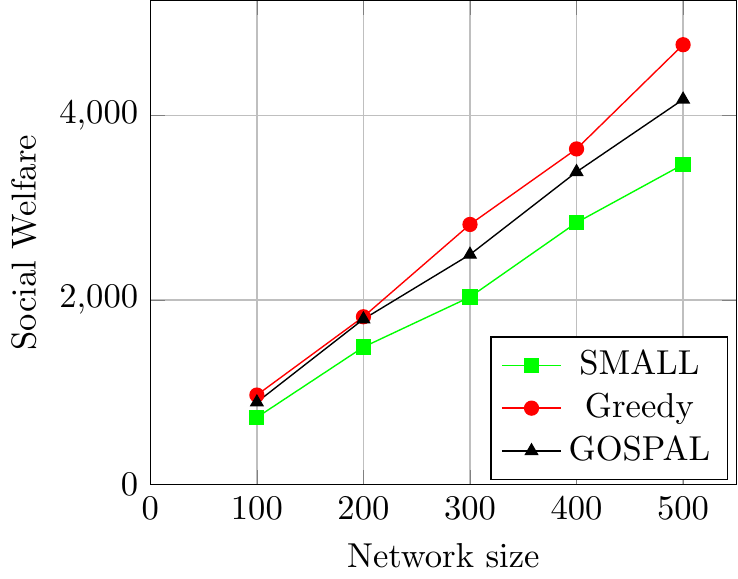}
        \caption{{Social Welfare}}
        \label{fig:sw1}
        \end{subfigure}%
    ~
    \begin{subfigure}[t]{0.28\textwidth}
        \centering
       \includegraphics[width=\textwidth]{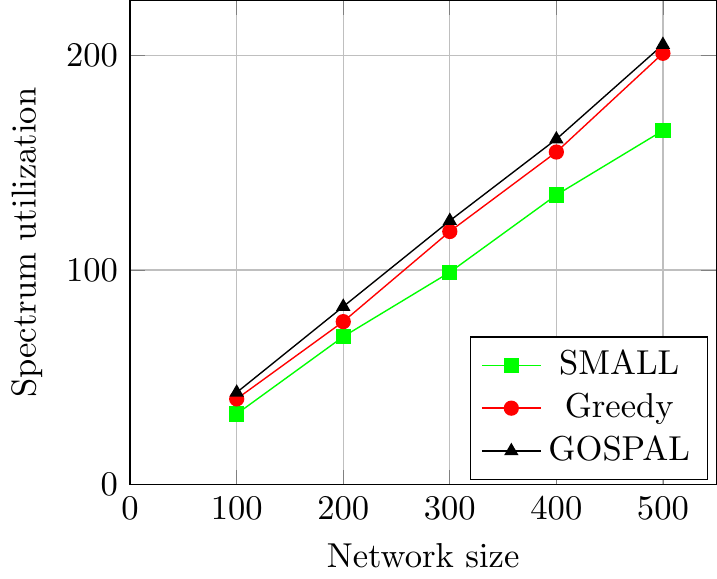}
        \caption{{Spectrum utilization}}
        \label{fig:su1}
        \end{subfigure}
     ~
         \begin{subfigure}[t]{0.28\textwidth}
        \centering
       \includegraphics[width=\textwidth]{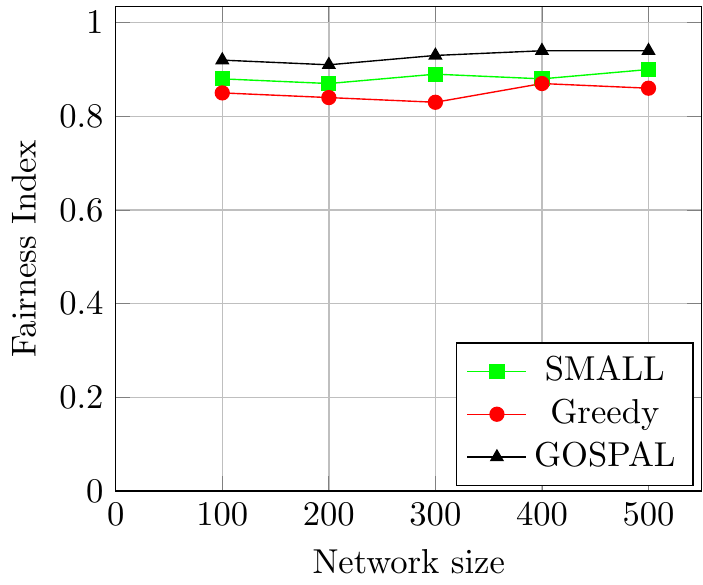}
        \caption{{Fairness Index}}
        \label{fig:fair_idx2}
        \end{subfigure}
\caption{{Performance comparison for different algorithms in large network.}}
\label{compare_2}
\end{figure*}

\begin{figure*}[t!]
    \centering
    \begin{subfigure}[t]{0.28\textwidth}
        \centering
       \includegraphics[width=\textwidth]{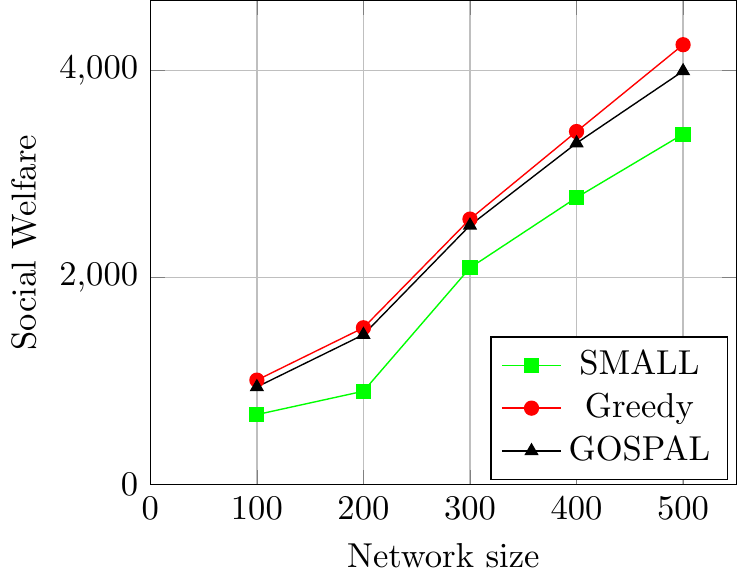}
        \caption{{Social Welfare}}
        \label{fig:sw1}
        \end{subfigure}%
    ~
    \begin{subfigure}[t]{0.28\textwidth}
        \centering
       \includegraphics[width=\textwidth]{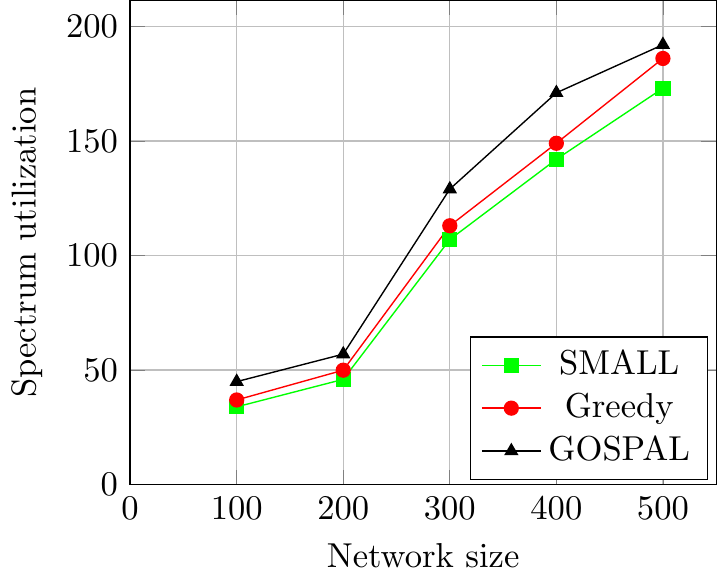}
        \caption{{Spectrum utilization}}
        \label{fig:su1}
        \end{subfigure}
     ~
         \begin{subfigure}[t]{0.28\textwidth}
        \centering
       \includegraphics[width=\textwidth]{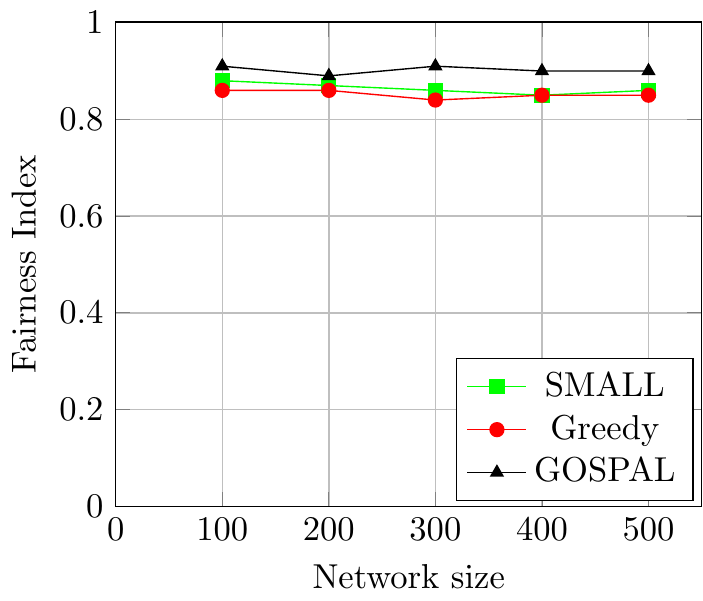}
        \caption{{Fairness Index}}
        \label{fig:fair_idx3}
        \end{subfigure}
\caption{{Performance comparison for different algorithms with increasing probability to higher degree interfering base station.}}
\label{compare_3}
\end{figure*}

\begin{figure*}[t]
    \centering
    \begin{subfigure}[t]{0.28\textwidth}
        \centering
       \includegraphics[width=\textwidth]{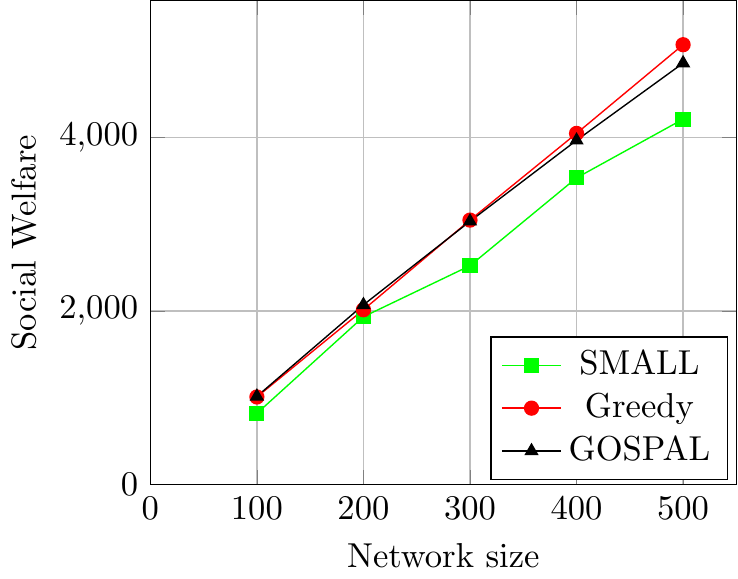}
        \caption{{Social Welfare}}
        \label{fig:sw1}
        \end{subfigure}%
    ~
    \begin{subfigure}[t]{0.28\textwidth}
        \centering
       \includegraphics[width=\textwidth]{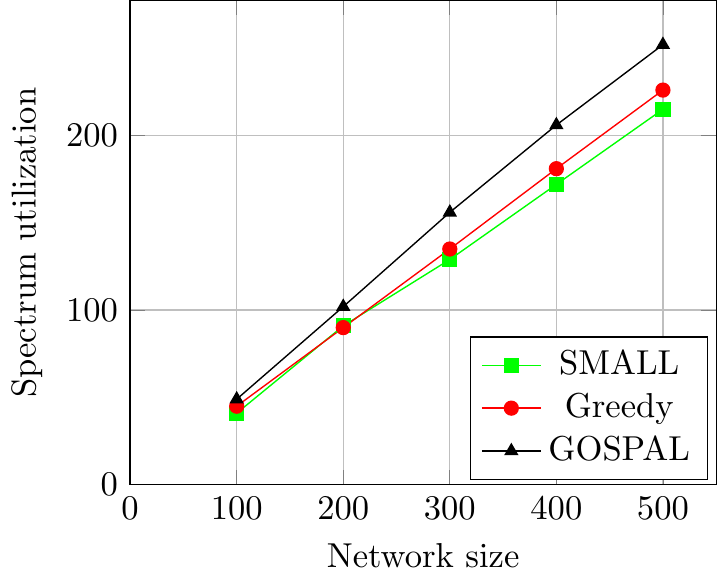}
        \caption{{Spectrum utilization}}
        \label{fig:su1}
        \end{subfigure}
     ~
         \begin{subfigure}[t]{0.28\textwidth}
        \centering
       \includegraphics[width=\textwidth]{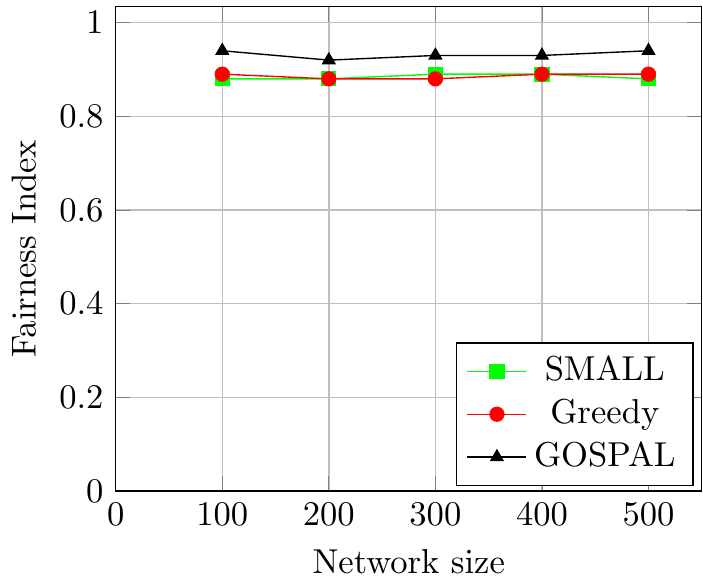}
        \caption{{Fairness Index}}
        \label{fig:fair_idx4}
        \end{subfigure}
\caption{{Performance comparison for different algorithms with decreasing probability to higher degree interfering base station.}}
\label{compare_4}
\end{figure*}

Next to understand the impact of the degree distribution on the performance of various schemes, we repeat the same experiment as above
with following probability mass functions over degree values $\{1,2,3,4\}$: (a)~$(0.1,0.2,0.3,0.4)$ and (b)~$(0.4,0.3,0.2,0.1)$.
Note that in case (a), the network will have large number of nodes with degree four, while in case (b) large number of nodes will have degree 1.
Figures~\ref{compare_3} and~\ref{compare_4} provide the results in case (a) and (b), respectively. Note that the results follow similar pattern as that in the
uniform degree case. These experiments demonstrate that the proposed algorithm though outperforms SMALL, only provides comparable performance with respect to greedy.

When a large percentage of base stations within the network have a high degree of conflict, both social welfare and resource utilization are reduced. This is illustrated in Figure~\ref{compare_3}. On the other hand, as shown in Figure~\ref{compare_4}, a reduction in the percentage of base stations with a high degree of conflict sets, leads to improvements in the above mentioned parameters. This reflects the fact that lesser number of base stations are allocated channels if the constraint set $\cS_i$ is large and vice versa.

\begin{figure}
\centering
\includegraphics[width = 0.35\textwidth]{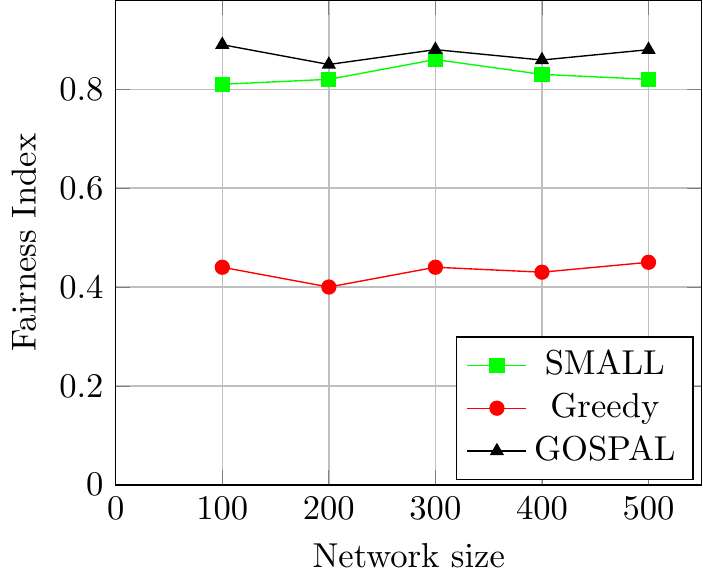}
\caption{ Comparison of fairness index for different algorithms in large networks.}
\label{fig:fairness}
 \end{figure}
 
Next, we perform simulations to see how various algorithms perform when the resource allocation process is repeated periodically. 
For this, we generate a random network topology and keep it fixed. In each topology we consider 100 different randomly generated bid values.
For each, we calculate the resource allocation under all the three schemes. Let $\alpha_i^\pi$ denote the fraction of time base station $i$ is allocated
resource under mechanism $\pi$. Based on the vector $\alpha_i^\pi$ we calculate Jain's fairness index. Jain's fairness index is a metric used in networking to
determine the share of system resources allocated to a user. 
First, we consider a case when bid values are independent and identically distributed (iid) across users and time.
The fairness for various schemes is shown in Figures~\ref{fig:fair_idx1},~\ref{fig:fair_idx2},~\ref{fig:fair_idx3},~\ref{fig:fair_idx4}. Note that all the schemes perform similarly.
Second, we consider a scenario in which bid values are iid across time, but not across users. For this, we initially choose
a value $\mu_i$ uniformly distributed in interval [8,35] for each user $i$. Now the bid for user $i$ is generated in slot $t$ as
$\mu_i + b(t)$ where $b(t)$'s are iid and uniformly distributed in the interval [-2,1]. The fairness index for various schemes
is shown in Figure~\ref{fig:fairness}. In this case, GOSPAL and SMALL significantly outperform the greedy scheme.